\documentclass[aps,superscriptaddress,reprint
]{revtex4-1}

\usepackage{amsthm,amsmath,physics,mathtools,amssymb}

\usepackage{charter}

\usepackage{hyperref}
\hypersetup{
    colorlinks=true,
    linkcolor=blue,
    filecolor=magenta,      
    urlcolor=blue,
    citecolor=blue,
}

\usepackage{cleveref}

\usepackage{CJKutf8}

\usepackage{enumitem}

\makeatletter
\newtheorem*{rep@theorem}{\rep@title}
\newcommand{\newreptheorem}[2]{%
\newenvironment{rep#1}[1]{%
 \def\rep@title{#2 \ref{##1}}%
 \begin{rep@theorem}}%
 {\end{rep@theorem}}}
\makeatother

\newtheorem{theorem}{Theorem}
\newreptheorem{theorem}{Theorem}
\newtheorem{corollary}[theorem]{Corollary}
\newtheorem{lemma}[theorem]{Lemma}
\newtheorem{proposition}[theorem]{Proposition}

\newtheorem{example}[theorem]{Example}

\theoremstyle{plain}
\newtheorem{convention}[theorem]{Convention}

\usepackage{graphicx}

\DeclarePairedDelimiterX{\inp}[2]{\langle}{\rangle}{#1, #2}

\usepackage{xparse}
\NewDocumentCommand\Hil{o}{%
  \IfNoValueTF{#1}
   {\mathcal{H}}
   {\mathcal{H}^{#1}}
}
\NewDocumentCommand\LH{o}{%
  \IfNoValueTF{#1}
   {\mathcal{L}(\mathcal{H})}
   {\mathcal{L}(\mathcal{H}^{#1})}
}
\NewDocumentCommand\HH{o}{%
  \IfNoValueTF{#1}
   {\mathcal{L}^*(\mathcal{H})}
   {\mathcal{L}^*(\mathcal{H}^{#1})}%
}

\newcommand{\pp}[1]{#1^\perp}

\newcommand{\uh}{\hat{u}}
\newcommand{\ein}{\ \overline{\in}\ }

\newcommand\id{\leavevmode\hbox{\small1\kern-3.3pt\normalsize1}}

\begin{document}

\begin{CJK*}{UTF8}{gbsn}
\title{Correlational quantum theory and correlation constraints}
\author{Ding Jia (贾丁)}
\email{ding.jia@uwaterloo.ca}
\affiliation{Department of Applied Mathematics, University of Waterloo, Waterloo, Ontario, N2L 3G1, Canada}
\affiliation{Perimeter Institute for Theoretical Physics, Waterloo, Ontario, N2L 2Y5, Canada}

\begin{abstract}
A correlational dialect is introduced within the quantum theory language to give a unified treatment of finite-dimensional informational/operational quantum theories, infinite-dimensional relativistic quantum theories, and quantum gravity. Theories are written in terms of correlation diagrams which specify correlation types and weights. Grouping similar correlation diagrams leads to generalized Feynman diagrams, which in special cases reduce to the familiar Feynman diagrams from quantum field theories. The correlational formalism is applied in a study of correlation constraints, revealing new classes of quantum processes that evade previous characterizations of general quantum processes. The results apply to quantum theories of various kinds, including time-asymmetric theories, time-symmetric theories, theories without predefined time, and theories with indefinite causal structures.
\end{abstract}

\maketitle
\end{CJK*}

\section{Introduction}

Quantum theory has gone through several phases of evolution. It started as the quantum mechanics of particles. Then came the quantum theory of fields. More recently the quantum theory of information has been on the rise. Will the future reveal yet new phases of quantum theory?




Our vision is that besides quantum particles, fields, and bits (dits), \textit{quantum correlations} should also be used as a fundamental concept in constructing quantum theories. Here we understand quantum correlation in a both general and specific way -- general because we hope to build correlational theories that break the boundaries among theories based on quantum particles, fields, and bits (dits), and specific because we want to offer a concrete prescription for constructing correlational theories.

Generally, we understand quantum correlation as \textit{anything that is mediated and has a quantifiable weight} in a quantum theory. At this level, the notion of quantum correlation is intentionally abstract so that it transcends the distinction between particles and fields, which both involve mediated quantifiable correlations, and go beyond qudits, which are limited to finite dimensions. 
In particular theories, the substance (e.g., particles, fields etc.) that mediates the quantum correlations will be specified, and the abstract notion of quantum correlation will be supplemented with more concrete physical connotations.

Specifically, we offer a prescription to construct correlational quantum theories based on ``correlation diagrams'', which are graphs labelled by correlation types and correlation weights. Upon composition, new diagrams form out of old diagrams, correlations mediate selectively according to type (color) matching, and new weights are calculated according to old weights. 
The essence of this correlational formalism is to keep track of the mediation of correlations in a manifest way through the composition of diagrams. Previous works on quantum theory that hold correlations essential include \cite{Rovelli1996RelationalMechanicsb, Mermin1998TheMechanics, *Mermin1998WhatUsb, HardyProbabilityGravity, *hardy2007towards, KempfTheStatistics}. 
Previous works that express quantum theory as a compositional/diagrammatic theory include \cite{Feynman1949Space-TimeElectrodynamics, *tHooft1974Diagrammarb, Coecke2004TheEntanglement, *coecke2014TheEntanglement, *abramsky2009categorical, *coecke2017picturing, oeckl2003general, *oeckl2013positive, *Oeckl2019APhysics, Chiribella2010ProbabilisticPurification, *Chiribella2011InformationalTheory, hardy2011foliable, *Hardy2012TheTheory, *Hardy2013OnPhysics}. The correlation diagrams proposed here are in particular similar to Hardy's duotensors. 

As we show in the paper, the diagrammatic setup considered here allows the construction of both finite-dimensional informational/operational quantum theories and infinite-dimensional relativistic quantum theories. It enable us to transport concepts and tools among theories based on quantum particles, field and bits (dits). 

As an interesting example, we show how the notion of Feynman diagrams can be generalized to represent classes of correlation diagrams for both finite and infinite-dimensional theories, with or without perturbative considerations. In relativistic theories the generalized Feynman diagrams reduce to the familiar Feynman diagrams. From this perspective, Feynman diagrams are not merely mathematical bookkeeping devices restricted to field theories or perturbation considerations, but arise generically when quantum correlational configurations are under investigation.

The correlational formalism also supplies useful technical tools. As an example, we introduce a binary string calculus to study correlation constraints under composition and decomposition. This study yields new classes of processes that go beyond previous characterizations of the most general finite-dimensional quantum processes.

Finally, the correlational perspective leads to a relational approach to quantum gravity, as developed in a separate paper \cite{JiaWorldFunction}.

\section{Correlational quantum theories}\label{sec:cd}

Given a multipartite quantum channel $N$, e.g., in its Kraus operator description $N:\rho\mapsto \sum_i K_i\rho K_i^\dagger$, how do we know which parties are correlated?

\begin{figure}
    \centering
    \includegraphics[width=.4\textwidth]{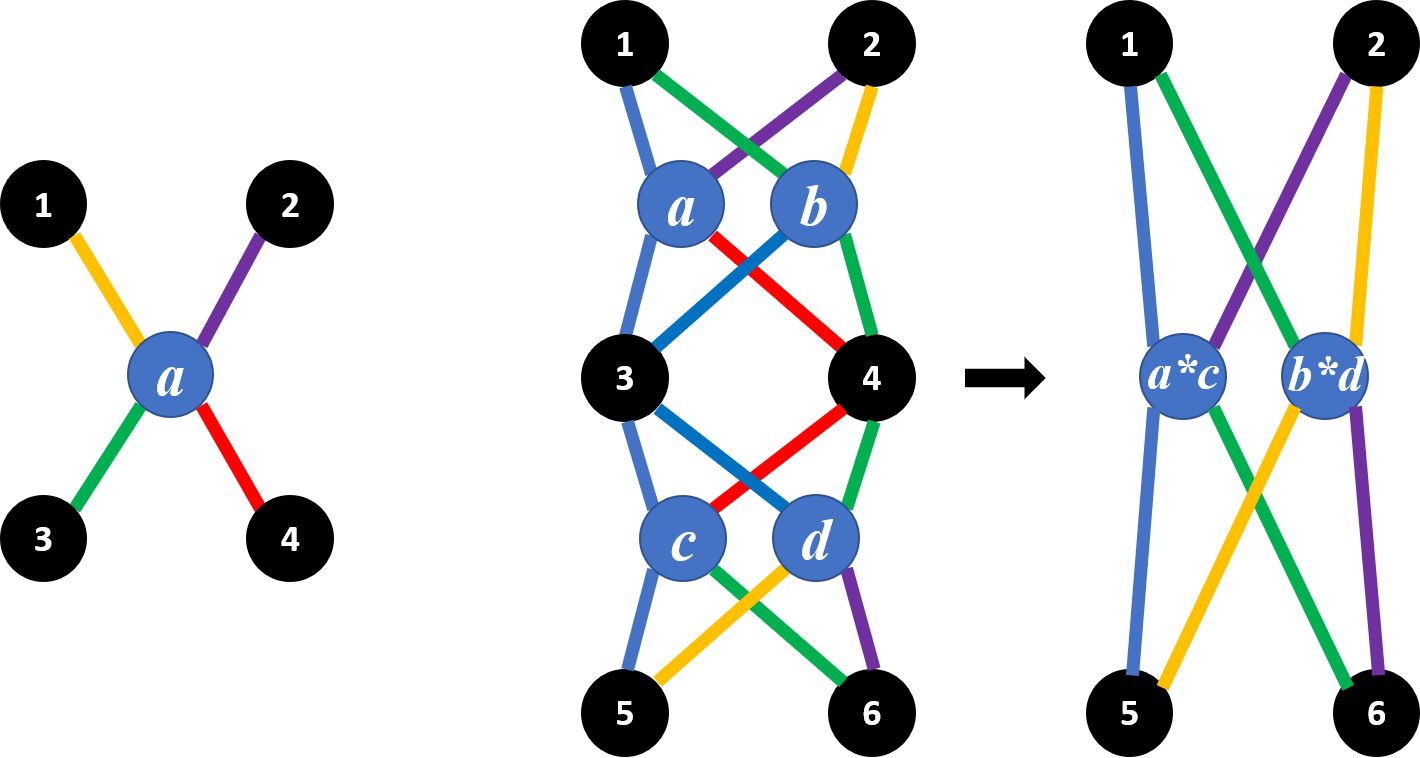}
    \caption{(colored online) Left: A correlation diagram. The vertices represent systems, the colored legs represent correlation types, and the letter represents variables quantifying correlation weights.
    Right: Correlation diagram composition. Compose the pairs $(a,c)$, $(a,d)$, $(b,c)$ and $(b,d)$ at systems $3$ and $4$. The $(a,c)$ and $(b,d)$ pairs have matching correlation types at the composed systems, so the correlation gets mediated through. The new weight variables $a*c$ and $b*d$ are calculated by some yet unspecified theory-dependent algorithms. The $(a,d)$ and $(b,c)$ pairs have mismatching correlation types, so are eliminated. 
    }
    \label{fig:cd}
\end{figure}

To make manifest the correlational structure of quantum processes, we introduce \textbf{correlation diagrams} to represent correlational configurations. 
As illustrated in \Cref{fig:cd}, a correlation diagram is a graph with colored edges representing \textbf{correlation types}, and a weight variable representing \textbf{correlation weights}. Correlation diagrams compose to represent the mediation of correlations. Old diagrams give rise to new diagrams only if the correlation types match at the composed systems. The new weight variable is obtained by theory-specific rules from the old weight variables. This setting captures correlation as something that can be mediated/blocked (according to types) and quantified (according to weight variables).

A process such as a quantum channel is described by a list of correlation diagrams, and only systems/parties sharing diagrams with non-trivial types are correlated (which shall be clear form the study on correlation constraints below).



The correlational formalism allows one to \textit{define} theories directly in terms of correlation diagrams from the outset without going through Hilbert space vectors and operators, much like the path integral formalism \footnote{Indeed, the correlational formalism may be viewed as a \textit{compositional path integral} approach where correlational configurations obeying certain composition rules are summed over.}. 

The correlational formalism is quite versatile. It is applicable to theories with both finite and infinite dimensional systems, as explained below in this section. In addition, it is applicable both at the complex amplitude level (e.g., the level of state vectors) and the real probability level (e.g., the level of density operators). 

\subsection{Finite-dimensional theories}\label{sec:fdt}

In operational terms the basic elements of a quantum theory are preparation, evolution, and measurement. These can all be represented as completely positive (CP) maps \cite{haag1964algebraic, davies1970operational}, taking possibly the trivial one-dimensional systems as inputs and/or outputs \cite{Chiribella2010ProbabilisticPurification, *Chiribella2011InformationalTheory}. Quantum theory can then be formulated in terms of composition of CP maps \cite{chiribella2009theoretical}.

A convenient correlation diagram description of a CP map $A$ is the following.
\begin{enumerate}
\item Obtain the Choi operator \cite{choi1975completely} (\Cref{sec:co}) of the CP map.
\item Expand the Choi operator in a generalized Pauli operator basis (\Cref{sec:gpo}).
\item The basis indices and expansion coefficients correspond to correlation types and correlation weight variables.
\end{enumerate}
The first two steps yields
\begin{align}\label{eq:gpbd}
\mathsf{A}=\sum_{i,j,\cdots, k} a_{ij\cdots k} \sigma_{i}^{1}\otimes \sigma_j^2\otimes\cdots\otimes \sigma_k^{n}
\end{align}
as Choi's positive semidefinite operator of the CP map $A$ associated to $n$ input and output systems. $\sigma^m_i$ is the $i$-th generalized Pauli operator on the $m$-th system, $a_{ij\cdots k}\in \mathbb{R}$ is a coefficient in the basis expansion, and the sum is over all the basis elements of all systems. The third step yields a list of correlation diagrams, with $i,j,\cdots, k$ represented as colored legs on the $1,2,\cdots, n$-th systems, and $a_{ij\cdots k}$ as the weight variables.

The correlation diagram composition rule follows from that of the Choi operators
\cite{chiribella2009theoretical}
\begin{align}\label{eq:cf}
\mathsf{A}*\mathsf{B}:=\frac{1}{d_{\mathcal{S}_*}}\Tr_{\mathcal{S}_*}[\mathsf{A}^{T_{\mathcal{S}_*}}\mathsf{B}].
\end{align}
$\mathcal{S}_*$ is the set of composed systems. $T_{\mathcal{S}_*}$ is the partial transpose with respect to a standard basis on $\mathcal{S}_*$. $\Tr_{\mathcal{S}_*}$ and $d_{\mathcal{S}_*}$ are the partial trace on and dimension of $\mathcal{S}_*$.
Since on every system $\Tr[\sigma_j^T \sigma_i]=\pm 2\delta_{i,j}$ (\Cref{sec:gpo}), all correlation types must match at the composed systems to generate a non-vanishing new diagram, and the new weight variable is the product of the old ones, multiplied by $(\prod_{m\in \mathcal{S}_*} c_m)/d_{\mathcal{S}_*}$, where $c_m=\pm 2$ originates from $\Tr[\sigma_j^T \sigma_i]=\pm 2\delta_{i,j}$ (in the convention of \Cref{sec:gpo}, $-2$ for $i=(j,k)$ with $1\le k<j \le d$, and $+2$ otherwise) and $d_{\mathcal{S}_*}$ from (\ref{eq:cf}). 

The correlation diagram description applies to any Hermitian operator description of processes obeying the composition rule (\ref{eq:cf}). In particular, it applies when processes are not given as CP maps but directly in terms of Hermitian operators (e.g., quantum theories without predefined time \cite{oreshkov2015operational, oreshkov2016operational}). 


The correlation diagrams represent not just the causal \textit{propagation} of correlations, but also the acausal \textit{mediation} of correlations. The correlation diagrams can be used to do (de)composition in the spacelike direction (e.g., decomposing an entangled state into two Hermitian operators), in addition to the timelike direction. Spacelike (de)composition are generically present in relativistic quantum theories (see below), and the current setup provides a finite-dimensional analogue. 


\subsection{Generalized Feynman diagrams}

Feynman diagrams are usually understood as bookkeeping symbols for a perturbation series. Here we present an alternative understanding of Feynman diagrams as classes of similar correlation configurations. They arise in finite- and infinite-dimensional quantum theories, with or without perturbative considerations.

\begin{figure}
    \centering
    \includegraphics[width=.49\textwidth]{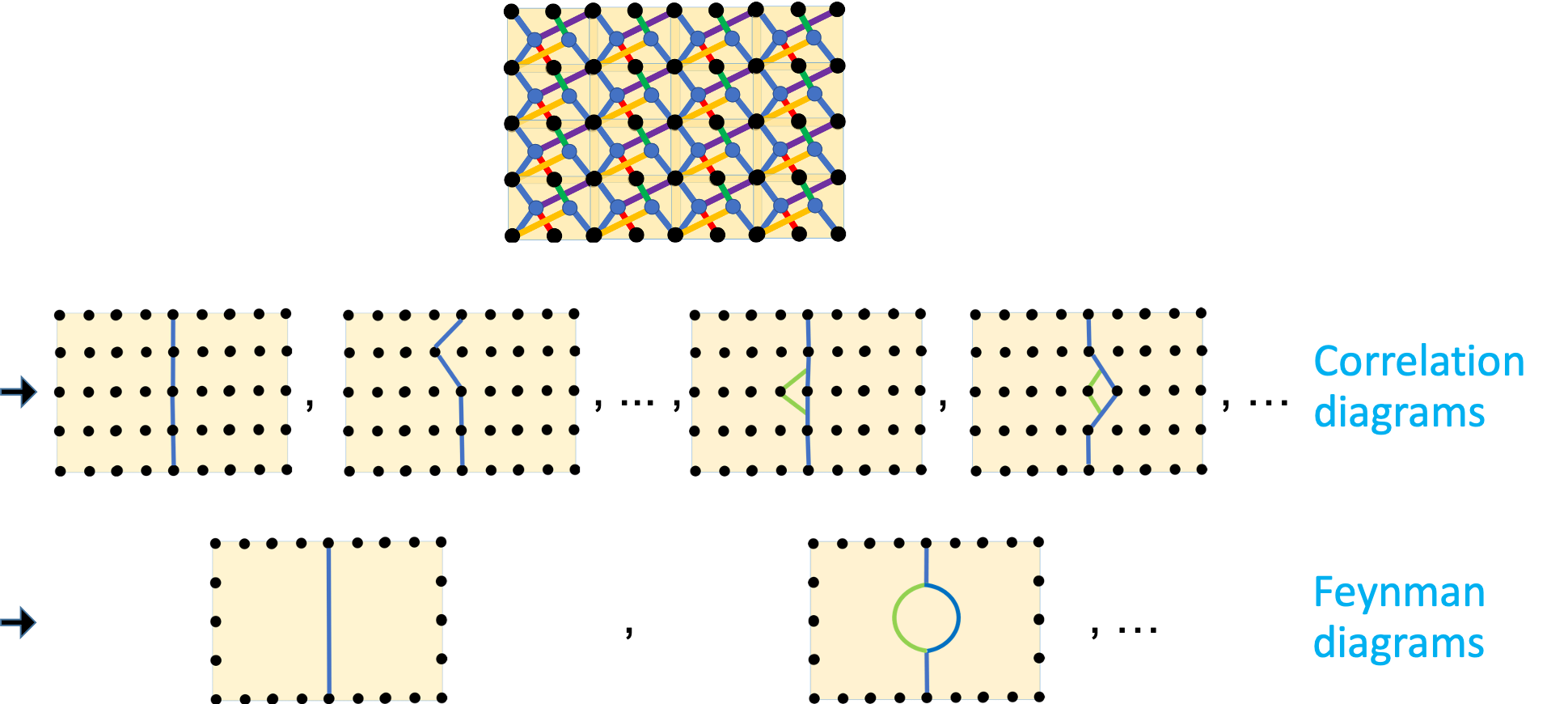}
    \caption{(colored online) Generalized Feynman diagrams as classes of correlation diagrams. In drawing we suppressed weight variables for simplicity.}
    \label{fig:fd}
\end{figure}

Consider the composition of multiple processes, such as the case depicted in \Cref{fig:fd}. Type matching leads to a list of correlation diagrams, such as those in the second line. 
In principle we can keep this list as the description of the composed process, but it is usually possible and convenient to reduce the number of correlational diagrams by grouping similar ones. For instance, for the theories defined in \Cref{sec:fdt} only the correlation types of the edges touching the boundary matter for subsequent compositions with other processes. For these theories, all four diagrams explicitly shown on the second line of \Cref{fig:fd} are of the same type. According to (\cref{eq:gpbd}), their weights add up linearly, so we can group them all together into a single correlation diagram. For a general case diagrams other than these four can also belong to the same type, and we may want to group all such diagrams together. For certain theories such as the relativistic models presented in the next subsection, when enumerating all such diagrams it is convenient to perform some intermediate steps of grouping as illustrated on the third line of \Cref{fig:fd}. The first two diagrams shown on the second line contain edges of the same type that trace out a line. They belong to a class represented by the first diagram on the third line. The other two diagrams shown on the second line contain two sets of edges of the same type that trace out lines, and the two lines touch each other. They belong to a class represented by the second diagram on the third line. 

We refer to diagrams representing classes of similar correlational configurations as \textbf{generalized Feynman diagrams} of the correlational formalism, because they reduce to the familiar Feynman diagrams of relativistic quantum theories in special cases. As noted ahead of \Cref{sec:fdt}, the correlational formalism is applicable at both the complex amplitude and the real probability levels, so the correlational Feynman diagrams could be at either of these levels. In addition, the correlational Feynman diagrams arise in situations both with and without perturbative considerations. As shown next, applied to certain relativistic quantum theories the correlational Feynman diagrams at the amplitude level correspond exactly to the familiar Feynman diagrams of quantum field theory.

\subsection{Relativistic quantum theories}

Relativistic quantum physics is usually studied in the quantum field theory paradigm based on pointwise field variables. 
The correlational formalism offers an alternative based on correlational configurations, which involves extended objects such as lines and surfaces. In this subsection we briefly sketch the main ideas, using free real scalar degrees of freedom as an example. Detailed account of this and other types of theories will be reported elsewhere.



A correlational configuration in a spacetime region is specified by a number of worldlines. Each worldline is either a closed loop, or an open line with ends attached to the boundary of the region. The boundary correlation types are specified pointwise by the number of worldlines attached to the point. Since the worldline number can be any natural number, we are dealing with infinite-dimensional systems. The process of the region is described by the list of all such correlational configurations. When composing different regions, type matching means only configurations whose worldlines numbers match pointwise on the boundary generate new configurations.

Once we have fixed a boundary condition on the region of interest to represent a physically meaningful event, we can group correlational configurations into classes. A generalized Feynman diagram contains open lines to match the boundary condition, plus a number of closed loops. An open line represents a sum over open worldlines matching the same boundary condition, and a closed loop represents a sum over closed worldlines. The lines and loops of the Feynman diagram are not located to any particular curve in spacetime, while the worldlines are. In this sense, the Feynman diagrams are topological classes of localized worldline configurations.

Note that the worldlines summed over can have spacelike parts and can move back and forth in time. We think of these as configurations of the (possibly spacelike) mediation of correlations in spacetime. Thinking in terms of real correlations is more straightforward than in terms of virtual particles that defy causal structures. 

It has been known since Feynman \cite{Feynman1950MathematicalInteraction, *Feynman1951AnElectrodynamics, *Polyakov1987GaugeStrings, *Bern1991EfficientAmplitudes, *Strassler1992FieldActionsb, *SchmidtTheGraphs, *CorradiniSpinningTheory, *Kleinert2009PathMarkets, *Costello2011RenormalizationTheory} that such sums over worldlines reproduce the Feynman field propagators and hence the quantum field theory Feynman diagrams. The correlational model outlined above then forms a relational analogue of free scalar field theory.

To formalize the model we need to give a prescription on enumerating the worldline configurations. This can be done with the help of lattices. Complex scalar, fermion, abelian and non-abelian gauge degrees of freedom can be incorporated with oriented worldlines, worldlines with internal spaces, worldsheets without and with internal spaces. Quantum gravity can be incorporated with the approach of \cite{JiaWorldFunction} or other relational graph-based approaches. These degrees of freedom when coupled together give a rich structure to the correlational configurations. The models at the complex amplitude level can also be developed to the real probability level along the lines of \cite{oeckl2013positive, *Oeckl2019APhysics}. The details of these constructions will be reported elsewhere. 

\section{Correlation constraints}\label{sec:cc}

We now focus on finite dimensions and conduct a study on characterizing constraints of correlations. We introduce a binary string calculus that turns out to simplify deductions and calculations. As a result, we identify new classes of quantum processes that fall beyond previous characterizations of general quantum processes.

A correlation constraint serves to tell which subsystems are allowed to be correlated. we write a constraint as a list of binary strings. For instance, consider channels from systems $s_1s_2$ to systems $s_3s_4$ constrained by $\{1010,0101,1011,0000\}_{s_1s_2s_3s_4}$. Correlations are allowed among systems sharing $1$'s. In this example correlations are allowed among $s_1s_3$, $s_2s_4$, and $s_1s_3s_4$, meaning $s_1$ can signal to $s_3$, $s_2$ can signal to $s_4$, and $s_1$ can signal to $s_3s_4$ if there is a bipartite decoding. Technically, the binary strings constrain correlation diagrams by correlation types. In terms of the generalized Pauli basis used in (\ref{eq:gpbd}), $\sigma_i\propto \id$ corresponds to $0$, and all other $\sigma_i$ correspond to $1$. Then for instance $0000$ allows terms of the type $\id^{s_1}\otimes\id^{s_2}\otimes\id^{s_3}\otimes \id^{s_4}$, while $1010$ allows terms of the type $\sigma_i^{s_1}\otimes \id^{s_2}\otimes \sigma_j^{s_3}\otimes \id^{s_4}$ for $\sigma_i, \sigma_j\not\propto \id$. A correlation constraint as a list of binary strings constrains the processes to only have terms of the types in the list.

The all zero string is special because it is present in all correlation constraints for physical processes. The Choi operators for physical processes are positive-semidefinite and non-zero, so must have positive trace. Since only terms corresponding to the zero string have non-zero trace (\Cref{sec:gpo}), the all zero string must be present. We give it a special symbol
\begin{align}
u=00\cdots 0.
\end{align}
Second, processes that preserve probabilities have a fixed term $\id$ for the type $u$. We introduce a special symbol $\uh$ to correspond to $\id$. For example, $\{1010,0101,1011,\uh\}_{s_1s_2s_3s_4}$ implies the only term of type $u$ the processes can have is the fixed term $\id^{s_1}\otimes\id^{s_2}\otimes\id^{s_3}\otimes \id^{s_4}$.
Since $\id$ obviously belongs to the type $u$, we assume that for correlation constraints, $\uh$ is implicitly present whenever $u$ is present. When $\uh\in A$ but $u\notin A$, we write
\begin{align}\label{eq:ein}
 \uh\ein A.
\end{align}


The binary strings together with $\uh$ are called \textbf{correlation type elements}, which are referred to using lower-case roman letters in the following. A set of correlation type elements is denoted by a capital roman letter such as $A$. The set of all processes (they have positive semidefinite Choi operators) constrained by $A$ is denoted by $\mathcal{C}_A$.

\subsection{Constraints under composition and decomposition}

We are interested in two kinds of questions regarding the constraints and the mediation of correlations:
\begin{align*}
a)\quad \mathcal{C}_{A}*\mathcal{C}_{B}\rightarrow ~?
\\
b)\quad \mathcal{C}_{A}*~? \rightarrow\mathcal{C}_{B}
\end{align*}
a) asks what we can say about the composition of processes constrained by $A$ and $B$, while b) asks what we can say about the decomposition of processes constrained by $B$ when one decomponent is constrained by $A$. First consider a).
\begin{theorem}\label{th:cmp}
For the composition on systems $\mathcal{S}_*$, 
\begin{align}\label{eq:cs}
&\mathcal{C}_{A}*\mathcal{C}_{B}\subset \mathcal{C}_{A*B}.
\end{align}
\end{theorem}
The proof is given in \Cref{sec:pcmp}. Here
\begin{align}\label{eq:btc}
A* B:=\{a+b:a\in A, b\in B, (a+b)_{{\mathcal{S}_{*}}}\in \{u,\uh\}\},
\end{align}
where $a+b$ is binary string addition (for $\uh$, $\uh+x=x+\uh=x$ for any $x$ including $\uh$), and $x_\mathcal{S}$ is $x$ restricted to systems $\mathcal{S}$ ($\uh_\mathcal{S}=\uh$). A Choi operator $\mathsf{A}$ defined on systems $\mathcal{S}_1$ is automatically extended to $\mathcal{S}_1\cup \mathcal{S}_2$ as $\mathsf{A}\otimes \id$ with identity acting on the new systems, so that the correlation type elements $a$ and $b$ are put on the same systems $\mathcal{S}_A\cup \mathcal{S}_B$ to carry out the addition and the restriction. As explained in \Cref{sec:pcmp}, \Cref{th:cmp} gives the best general characterization for composition, since $\mathcal{C}_{A*B}$ on the right hand side cannot be reduced, and $\subset$ cannot be strengthened to $=$.

Now consider b). For compositions on $\mathcal{S}_*$ so that $\mathcal{S}_*\cap \mathcal{S}_B=\varnothing$ (since composition eliminates systems), define 
\begin{align}
\mathcal{C}_{A}\rightarrow \mathcal{C}_{B}:=&\{\mathsf{C}\in\mathcal{C}:\mathcal{C}_{A}*\mathsf{C}\subset \mathcal{C}_{B}\},
\\
A\rightarrow B:=&\{c:A*c\subset B\}\label{eq:atb}
\end{align}
where $\mathcal{C}$ is the set of all positive semidefinite operators. $\mathcal{C}_{A}*\mathsf{C}=\{\mathsf{A}*\mathsf{C}:\mathsf{A}\in \mathcal{C}_{A}\}$, so that $\mathcal{C}_{A}\rightarrow \mathcal{C}_{B}$ contains all processes obeying the constraint of b). $A*c$ is understood according to (\ref{eq:cs}) where the second set has one correlation type element $c$. 
\begin{theorem}\label{th:catb}
For compositions on $\mathcal{S}_*$ so that $\mathcal{S}_*\cap \mathcal{S}_B=\varnothing$,
\begin{align}
\mathcal{C}_{A}\rightarrow \mathcal{C}_{B}&=\mathcal{C}_{A\rightarrow B},\label{eq:catbc}
\\
A\rightarrow B&=
\begin{cases}
(\pp{B}-A)_{\mathcal{S}_{A\rightarrow B}}^\perp\cup\{\uh\}, \quad &A*\uh\subset B,
\\
(\pp{B}-A)_{\mathcal{S}_{A\rightarrow B}}^\perp, \quad &\text{otherwise},
\end{cases}\label{eq:catb}
\\
\mathcal{C}_{A}\rightarrow \mathcal{C}_{B}&\ne \{0\} \text{ if and only if }A*\uh\subset B.\label{eq:ec}
\end{align}
\end{theorem}
The proof of the theorem is given in \Cref{sec:pcatb}. $\mathcal{S}_a$, the \textbf{support} of $a$, is the set of systems on which $a$ is not $0$. $\mathcal{S}_A:=\cup_{a\in A}\mathcal{S}_a$. $\mathcal{S}_{\uh}=\mathcal{S}_{u}=\varnothing$, and $\uh$ and $u$ are understood as supported on an one-dimensional trivial system. $\tau_\mathcal{S}$ is the set of all binary strings (not including $\uh$) on $\mathcal{S}$, and
\begin{align}\label{eq:perp}
\pp{A_\mathcal{S}}&:=\tau_{\mathcal{S}}\backslash A,&
\pp{A}&:=\tau_{\mathcal{S}_A}\backslash A,
\end{align}
with the conventions that $A$ or $\tau_\mathcal{S}$ be extended by joining $0$'s to carry out $\tau_{\mathcal{S}}\backslash A$ if necessary, $\pp{\{\uh\}}=\{u\}$, and $\pp{\{u\}}=\varnothing$ (think of $\tau_{\mathcal{S}}= \{u\}$ on the trivial system $\mathcal{S}=\varnothing$).
$A-B=\{a-b: a\in A, b\in B\}$, where analogous to $a+b$, $a-b$ is binary string subtraction under automatic extension. $a+b\ne a-b$ in general because while $a-\uh=a$ for all $a$ including $a=\uh$, $\uh-b$ outputs no element for $b\ne\uh$ (since only $\uh$ added to $\uh$ gives $\uh$).

The (\ref{eq:catb}) characterization for $A\rightarrow B$ admits an interpretation as excluding (through ``$_{\mathcal{S}_{A\rightarrow B}}^\perp$'') elements ($\pp{B}-A$) that compose with $A$ to form elements outside $B$. Finally $\uh$ needs special care when $A*\uh\subset B$.

Next we use \Cref{th:catb} to reproduce and generalize previous characterizations of general processes.


\subsection{Process matrices}

The process matrix framework \cite{oreshkov2012quantum, araujo2015witnessing, oreshkov2016causal} takes an operational approach to study general quantum correlations. Start with $n$ ``local laboratories'' inside each an agent performs an \textit{arbitrary} quantum operation, modelled as a quantum instrument \cite{davies1970operational} from one input to one output system. The $n$-party quantum correlations are encoded in the process matrices as Choi operators that always yield valid probabilities (non-negative probabilities that sum to one) when composed with these arbitrary local operations. The most general such correlations can indicate the presence of quantum indefinite causal structure among the agents, so generalize ordinary quantum theory with definite causal structure. Precisely which of these theoretically defined process matrices are realizable in nature is an open question.

The valid probability requirement imposes correlation constraints. For example, for two parties it imposes the constraint $W_2=\{\uh,1000,0010,1010,0110,1001,1110,1011\}$ \cite{oreshkov2012quantum}, which can be reproduced by applying Theorem \ref{th:catb} to the following setup.
\begin{example}[Two-party process matrices]
$\mathcal{S}_*=\{s_1,s_2,s_3,s_4\}.$ $A=\{\uh, 01, 11\}_{s_1s_2}\times\{\uh, 01, 11\}_{s_3s_4}$. $B=\{\uh\}$.
\end{example}
\noindent 
The first agent has input $s_1$ and output $s_2$, and the operation is constrained by $\{\uh, 01, 11\}_{s_1 s_2}$ to be a channel. Similarly for the second agent. $A$ denotes the constraint on the joint operations as product channels. ($X\times Y:=\{xy:x\in X, y\in Y\}$ denotes elongated elements so that $\mathcal{C}_X\otimes \mathcal{C}_Y=\mathcal{C}_{X\times Y}$. E.g., if $a=00, b=11$, then $ab=0011$. For $\uh$, $\uh y=u y$, $x\uh=x u$, and $\uh \uh=\uh$.) 
$B$ on the trivial system constrains the probability to be $1$. $A\rightarrow B$ then defines the process matrices as composing with arbitrary product channels to yield the normalized probability. 
Since $\pp{B}-A=\{u\}-A=(A\backslash\{\uh\})\cup\{u\}$, Theorem \ref{th:catb} implies $A\rightarrow B=(\pp{B}-A)_{\mathcal{S}_{A\rightarrow B}}^\perp\cup\{\uh\}=(\pp{A}\backslash\{u\})\cup\{\uh\}=W_2$

The $n$-parties constraint \cite{araujo2015witnessing, oreshkov2016causal} can similarly be reproduced by applying Theorem \ref{th:catb} to the following setup.
\begin{example}[$n$-party process matrices]
$\mathcal{S}_i=\{s_{2i-1},s_{2i}\}$. $\mathcal{S}_*=\cup_{i=1}^n \mathcal{S}_i=\{s_1,s_2,\cdots, s_{2n-1},s_{2n}\}$. $A_i=\{\uh,01,11\}$ with support $\mathcal{S}_{A_i}=\mathcal{S}_i$. $A=\times_i A_i$. $B=\{\uh\}$.
\end{example}
\noindent We have $A\rightarrow B=(\pp{B}-A)_{\mathcal{S}_{A\rightarrow B}}^\perp\cup\{\uh\}=(\pp{A}\backslash\{u\})\cup\{\uh\}=\{\uh\}\cup\{a\in \tau_{\mathcal{S}_*}: a_{\mathcal{S}_i}\in A_i^\perp\text{ for any }1\le i\le n\}=\{\uh\}\cup\{a\in \tau_{\mathcal{S}_*}: a_{\mathcal{S}_i}=10\text{ for any }1\le i\le n\}$.

The process matrices were originally defined under the assumption that the parties have one input and one output, and can apply arbitrary channels across its input and output \cite{oreshkov2012quantum, araujo2015witnessing, oreshkov2016causal}. What if the parties have multiple input and output subsystems, and their operations are constrained (e.g., by the spacetime causal structure in the laboratories)? A similar application of Theorem \ref{th:catb} addresses this question.
\begin{proposition}[Generalized $n$-party process matrices]
Let $\mathcal{S}_i$ be the set of systems of the $i$-th party, $\mathcal{S}_*=\cup_{i=1}^n \mathcal{S}_i$, $A_i$ constrain the allowed channels in the $i$-th party, $A=\times_{i}A_i$, and $B=\{\uh\}$. Then $A\rightarrow B=(\pp{A}\backslash\{u\})\cup\{\uh\}=\{\uh\}\cup\{a\in \tau_{\mathcal{S}_*}: a_{\mathcal{S}_i}\in A_i^\perp\text{ for any }1\le i\le n\}$.
\end{proposition}

\subsection{Higher order maps}

A higher order map  \cite{chiribella2009theoretical, perinotti2017causal, Bisio2019TheoreticalTheory, Castro-Ruiz2018DynamicsStructures,Kissinger2019AStructure} is one that maps between lower order maps. E.g., $A\rightarrow B$ characterizes higher order maps between those characterized by $A$ and $B$. This philosophy was used in \cite{chiribella2009theoretical, perinotti2017causal, Bisio2019TheoreticalTheory} to iteratively construct a hierarchy of maps specified by what this paper regards as correlation constraints. At the lowest level are states on single systems. Next comes maps between single system states etc. This hierarchy is quite general, incorporating all the process matrices considered before in the literature, which already goes beyond ordinary quantum theory with definite causal structure. Does this hierarchy capture all processes of interest?

From the correlational perspective there is no reason to focus on processes belonging to this hierarchy. Theorem \ref{th:catb} democratizes and generalizes the hierarchical higher order maps: 
1) In each application of the theorem, the multi-system correlation constraints $A$ and $B$ can be specified arbitrarily and need not come from any hierarchy. 2) $A$ and $B$ may share any common set of systems, in which case the composition is on a subset of the systems of $A$. 3) The processes that the theorem applies to need not distinguish input and output systems (e.g., as in the case of generalized Choi operators of \cite{oreshkov2016operational} without predefined time). 
As a special case, Theorem \ref{th:catb} implies:
\begin{corollary}
Let the constraints $A$ and $B$ be supported on distinct systems, i.e., $\mathcal{S}_A\cap\mathcal{S}_B=\varnothing$. For $\mathcal{S}_{*}=\mathcal{S}_A$ (recall (\ref{eq:ein}))
\begin{align}
A\rightarrow B=
\begin{cases}
(\pp{A}\times \tau_{\mathcal{S}_B})\cup (A\times B)\cup\{\uh\}, \quad & A*\uh\subset B,
\\
(\pp{A}\times \tau_{\mathcal{S}_B})\cup (A\times B), \quad &\text{otherwise}.
\end{cases}
\end{align}
\end{corollary}
\begin{proof}
Since $\mathcal{S}_A\cap\mathcal{S}_B=\varnothing$, $\pp{B}-A=A\times \pp{B}$. In addition, $\mathcal{S}_{A\rightarrow B}=\mathcal{S}_A\cup\mathcal{S}_B$. Then $(\pp{B}-A)_{\mathcal{S}_{A\rightarrow B}}^\perp=(A\times \pp{B})_{\mathcal{S}_A\cup\mathcal{S}_B}^\perp=(\pp{A}\times \tau_{\mathcal{S}_B})\cup (A\times B)$. The result follows from Theorem \ref{th:catb}.
\end{proof}
The set $(\pp{A}\times \tau_{\mathcal{S}_B})\cup (A\times B)$ can also be expressed as $(\pp{A}\times B)\cup(\pp{A}\times \pp{B})\cup (A\times B)$.
If $A$ and $B$ characterize normalized processes, i.e., $\uh\in A, B$, then the condition  $A*\uh\subset B$ is fulfilled, and the above result reproduces Proposition 5.6 of \cite{Bisio2019TheoreticalTheory} (see also Lemma 2 of \cite{perinotti2017causal}) that characterizes general hierarchical higher order maps. If $A$ and $B$ characterize process matrices, then the above result characterizes transformations of process matrices, and an iteration reproduces the characterization of \cite{Castro-Ruiz2018DynamicsStructures}.

\section{Conclusion}\label{sec:d}

We presented a correlational formalism that gives a unified treatment of informational/operational quantum theories and relativistic quantum theories. In a separate paper \cite{JiaWorldFunction}, we show in relational treatment that quantum gravity can also be incorporated. In all these formulations, correlation diagrams play a crucial role as configurations for the mediation of correlations in classical or quantum spacetime. Generalized Feynman diagrams emerge as classes of correlations for general finite and infinite-dimensional quantum theories, even when no perturbation is performed. Through studying correlation constraints, we found new classes of quantum processes that evade previous characterizations. An interesting topic for future research is to study measures of correlation weight based on the correlation diagrams, which are applicable across quantum information theory, relativistic quantum theories, and quantum gravity.

\section*{Acknowledgement}

I am very grateful to Lucien Hardy, Achim Kempf, Nitica Sakharwade, Fabio Costa, and Ognyan Oreshkov for valuable discussions. I especially thank Rafael Sorkin for his cautions on the informational approach to fundamental physics (partially summarized in the opening part of \cite{ProspectsApproaches}) -- this work is an attempt to reduce the gap between models and frameworks.

Research at Perimeter Institute is supported in part by the Government of Canada through the Department of Innovation, Science and Economic Development Canada and by the Province of Ontario through the Ministry of Economic Development, Job Creation and Trade. This publication was made possible through the support of the grant ``Causal Structure in Quantum Theory'' from the John Templeton Foundation and the grant ``Operationalism, Agency, and Quantum Gravity'' from FQXi. The opinions expressed in this publication are those of the author and do not necessarily reflect the views of the funding agencies. 

\appendix

\section{Choi operators}\label{sec:co}

For a Hilbert space $\Hil$, we denote by $\LH$ the space of bounded linear operators on $\Hil$. By the Choi isomorphism  \cite{choi1975completely}, there is an one-to-one correspondence between completely positive (CP) maps $\mathcal{M}:\mathcal{L}(\mathcal{H}^{a_1})\rightarrow \mathcal{L}(\mathcal{H}^{a_2})$ and positive-semidefinite operators $\mathsf{M}\in \mathcal{L}(\mathcal{H}^{a_2}\otimes\mathcal{H}^{a_1})$,
\begin{align}\label{eq:choi}
\mathsf{M}:=d_{a_2}(\mathcal{M}\otimes \mathcal{I}) \sum_{i,j=1}^{d_{a_1}}\ketbra{ii}{jj},
\end{align}
where $\mathcal{I}$ is the identity channel on system $a_1$, $d_x=\dim \mathcal{H}^{x}$, and the sums are over an orthonormal basis of $\mathcal{H}^{a_1}$. The normalization convention is so chosen that $M=\id+X$, where $X$ is a traceless operator. The positive-semidefinite operator in (\ref{eq:choi}) is called the \textbf{Choi operator} of the quantum process $\mathcal{M}$.

The above amounts to sending half of a (unnormalized) maximally entangled state to the original CP map to obtain a (unnormalized) bipartite state. For a CP map with multiple inputs and outputs, the Choi operator is obtained by sending half of a (unnormalized) maximally entangled state to each input.

A basic operation of CP maps is composition. Eventually the probabilistic predictions of the theory comes from composing processes. For example, composing a single measurement with a bipartite state leads to a reduced single system state, which when composed with another measurement leads to a list of probabilities for the measurement outcomes. For writing down the composition formula of Choi operators, it is convenient to automatically extend operators to larger sets of systems such that $\mathsf{A}$ acting on $\Hil^1$ (which may be a tensor product of Hilbert spaces) is freely viewed as $\mathsf{A}\otimes \id$ acting on $\Hil^1\otimes \Hil^2$ for arbitrary $\Hil^2$.

The composition on systems $\mathcal{S}_*$ of two operators $\mathsf{A}$ on systems $\mathcal{S}_\mathsf{A}$ and $\mathsf{B}$ on systems $\mathcal{S}_\mathsf{B}$ is given by the \textbf{composition formula} (\ref{eq:cf}) \cite{chiribella2009theoretical}
\begin{align}
\mathsf{A}*\mathsf{B}:=\frac{1}{d_{\mathcal{S}_*}}\Tr_{\mathcal{S}_*}[\mathsf{A}^{T_{\mathcal{S}_*}}\mathsf{B}],
\end{align}
where $T_{\mathcal{S}_*}$ is the partial transpose on $\mathcal{S}_*$ in the basis of the maximally entangled states used to obtain the Choi operator, and $\Tr_{\mathcal{S}_*}$ is the partial trace on $\mathcal{S}_*$. The normalization is so chosen to match ordinary probabilistic predictions. The composition symbol can be extended to sets of operators, so that $\mathcal{A}*\mathcal{B}:=\{\mathsf{A}*\mathsf{B}:\mathsf{A}\in\mathcal{A},  \mathsf{B}\in \mathcal{B}\}.$ 

The Choi operator has been generalized to processes that do not distinguish input and output systems \cite{oreshkov2016operational}. A processes can be specified directly in terms of a positive semidefinite operator instead of a CP map. The correlational formalism and the characterization of correlation constraints in this work apply to the generalized Choi operators as well, since we work directly with the positive semidefinite operators.

\section{Generalized Pauli operators}\label{sec:gpo}

On a $d$-dimensional Hilbert space $\Hil$, a set of $d^2$ many \textbf{generalized Pauli operators} $\sigma_i$ with $i=(m,n)$ for $m,n=1,\cdots, d$ forms a basis for the Hermitian operators \cite{Hioe1981NMechanics}:
\begin{widetext}
\begin{align}
\sigma_i=\sigma_{(m,n)}=
\begin{cases}
E_{mn}+E_{nm}, \quad &1\le m<n\le d,
\\
i (E_{mn}-E_{nm}), & 1\le n<m \le d,
\\
(\frac{2}{m(m+1)})^{1/2}(\sum_{l=1}^m E_{l}-mE_{m+1}), & 1\le m=n \le d-1,
\\
(\frac{2}{d})^{1/2} \id, & m=n=d,
\end{cases}
\end{align}
where $E_{mn}=\ketbra{m}{n}$, and $E_m=\ketbra{m}{m}$. $\Tr \sigma_i=0$ for all $i$ except $i=(d,d)$.

It is easy to check that
\begin{align}
\Tr[\sigma_i \sigma_j]=2\delta_{i,j}
\end{align}
so this is an orthogonal basis under the Hilbert-Schmidt norm for the real space of the Hermitian operators. In addition,
\begin{align}\label{eq:gpoc}
\Tr[\sigma_i \sigma_j^T]=\Tr[\sigma_j^T \sigma_i]=
\begin{cases}
-2\delta_{i,j}, \quad & i=(m,n), 1\le n<m \le d
\\
2\delta_{i,j} & \text{otherwise}.
\end{cases}
\end{align}
\end{widetext}
Hence orthogonality is preserved if one operator is transposed, as in the composition formula (\ref{eq:cf}). In the context of correlation diagrams, this means a color mismatch at any system of composition eliminates the diagrams.

\section{Correlation type expansion}

For a Hilbert space $\Hil$ with $\LH$ as the space of bounded linear operators on $\Hil$, denote by $L(\Hil)\subset \LH$ the real subspace of hermitian operators on $\Hil$. When $\Hil$ is clear from the context, we sometimes omit it and write $L$ for $L(\Hil)$. 

$L$ can be expanded into a traceful and a traceless part as
\begin{align}\label{eq:lds}
L=L_0\oplus L_1.
\end{align}
The traceful part $L_0\subset L$ is the (one-dimensional) subspace generated by the identity operator and the traceless part $L_1\subset L$ is the subspace of the traceless operators. The generalized Pauli operators of the previous section form a basis for $L$, with $\sigma_{(d,d)}$ spanning $L_0$ and the rest $\sigma_{(m,n)}$ spanning $L_1$.


On a tensor product space $\Hil=\Hil^1\otimes\Hil^2\otimes\cdots\otimes\Hil^m$, define
\begin{align}\label{eq:la}
L_{a}:=L_{a_1}\otimes L_{a_2}\otimes \cdots\otimes L_{a_m}\subset L(\Hil),
\end{align}
e.g., $L_{0100}=L_0\otimes L_1\otimes L_0 \otimes L_0$. Then
\begin{align}
L(\mathcal{H})=\oplus_{a} L_a,
\end{align}
where the sum is over all binary strings of length $m$. 
The binary strings together with $\uh$ form the correlation type elements. With $L_{\uh}:=\{\id\}$, each correlation type element now has a corresponding set of operators $L_a\subset L$. A set $A$ of correlation type elements corresponds to $L_A:=\oplus_{a\in A} L_a\subset L$. We set $L_\varnothing=\{0\}$. 



It follows that any operator $\mathsf{A}\in L(\mathcal{H})$ has a \textbf{correlation type expansion}
\begin{align}\label{eq:cte}
\mathsf{A}=\sum_{a\in A} \mathsf{A}_a, \quad 0\ne \mathsf{A}_a\in L_a,
\end{align}
where $\mathsf{A}_a$ are obtained by projecting $\mathsf{A}$ to $L_a$ and keeping non-zero elements. We use $\uh$ instead of $u$ in the set $A$ whenever possible. $A$ is called $\mathsf{A}$'s \textbf{correlation type}.

\section{Correlation type calculus}\label{sec:btc}

The correlation type/binary string calculus and related conventions are summarized here.

For the composition on systems $\mathcal{S}_*$, define
\begin{align}
A* B:=\{a+b:a\in A, b\in B, (a+b)_{{\mathcal{S}_{*}}}\in \{u,\uh\}\}.
\end{align}
Here $x_\mathcal{S}$ is $x$ restricted to systems $\mathcal{S}$, with $\uh_\mathcal{S}=\uh$.
\begin{convention}\label{cv:star}
Define
\begin{align}
A+B=\{a+b: a\in A, b\in B\}
\end{align} according to
\begin{align}\label{eq:apb}
a+b=
\begin{cases}
a+b\text{ as binary addition},\quad &a,b\ne \uh,
\\
b,\quad &a=\uh,b\ne\uh,
\\
a,\quad &b=\uh.
\end{cases}
\end{align}
$\{a\}*B$ and $\{a\}*\{b\}$ are often abbreviated as $a*B$ and $a*b$, respectively. Sometimes we abuse notation to treat $a*b$ as an element rather than a set and use expressions such as $a*b\in C$.
\end{convention}

Recall that a Choi operator $\mathsf{A}$ defined on systems $\mathcal{S}_1$ is automatically extended to $\mathcal{S}_1\cup \mathcal{S}_2$ as $\mathsf{A}\otimes \id$. Therefore the correlation type elements $a$ and $b$ can always be put on the same systems to carry out the addition, the restriction, and the following subtraction.

\begin{convention}\label{cv:cts}
Define 
\begin{align}
A-B=\{a-b: a\in A, b\in B\}.
\end{align} according to
\begin{align}\label{eq:amb}
a-b=
\begin{cases}
a-b\text{ as binary subtraction},\quad &a,b\ne \uh,
\\
\emptyset,\quad &a=\uh,b\ne\uh,
\\
a,\quad &b=\uh.
\end{cases}
\end{align}
Here $\emptyset$ means $a-b$ outputs no element, since nothing added to $b\ne \uh$ gives $a=\uh$. This $\emptyset$ symbol is used under the rule that $\emptyset+a=a+\emptyset=\emptyset-a=a-\emptyset=\emptyset$ for all $a$, and $\{\emptyset\}=\varnothing$.
\end{convention}

\section{Proof of \Cref{th:cmp}}\label{sec:pcmp}

\begin{lemma}\label{lm:lec}
For $\mathcal{S}_*\supset \mathcal{S}_a\cap \mathcal{S}_b$,
\begin{align}
L_a* L_b=L_{a*b}.
\end{align}
\end{lemma}
\begin{proof}
The composition formula $\mathsf{A}*\mathsf{B}=\frac{1}{d_{\mathcal{S}_*}}\Tr_{\mathcal{S}_*}[\mathsf{A}^{T_{\mathcal{S}_*}}\mathsf{B}]$  (\ref{eq:cf}) has three steps: 1) Partial transpose $\mathsf{A}$ at $\mathcal{S}_*$. 2) Multiply $\mathsf{A}^{T_{\mathcal{S}_*}}$ with $\mathsf{B}$. 3) Take the partial trace and multiply by $\frac{1}{d_{\mathcal{S}_*}}$. The first step leaves $L_a$ invariant, the second step forms product operators, and the third step with the partial trace projects out operators traceless on $\mathcal{S}_*$. 

Since the first step leaves $L_a$ invariant, we move to the second step and introduce the notation $L_a\cdot L_b:=\{\mathsf{A}\mathsf{B}:\mathsf{A}\in L_a, \mathsf{B}\in L_b\}$ for operator products. We claim that $L_{a+b}\subset L_a\cdot L_b\subset L_{a+b}\oplus L_C$, where $L_C$ is a space that will be projected out in the third step. First consider $a$ and $b$ on an individual system. Unless $a=b=1$, $L_a\cdot L_b=L_{a+b}$. When $a=b=1$, this system belongs to $\mathcal{S}_a\cap \mathcal{S}_b$ and hence also to $\mathcal{S}_*$. We have $L_{a+b}=L_0\subset L_a\cdot L_b\subset L_{a+b}\oplus L_C=L_0\oplus L_1$ with $L_{a+b}=L_0$ and $L_C=L_1$. $L_C$ will be projected out by the partial trace on $\mathcal{S}_*$. The generalization to multiple systems amounts to singling out the systems on which both $a$ and $b$ are $1$ and apply the same reasoning as above.

Now apply the third step to $L_{a+b}\subset L_a\cdot L_b\subset L_{a+b}\oplus L_C$ to get $L_{a*b}\subset L_a* L_b\subset L_{a*b}$ and hence the result
\end{proof}

\begin{lemma}
For $\mathcal{S}_*\supset \mathcal{S}_A\cap \mathcal{S}_B$,
\begin{align}\label{eq:cmp}
L_A*L_B\subset& L_{A*B}.
\end{align}
\end{lemma}
\begin{proof}
$L_A*L_B=(\oplus_{a\in A}L_a)*(\oplus_{b\in B}L_b)\subset \oplus_{a\in A,b\in B}L_a*L_b=\oplus_{a\in A,b\in B}L_{a*b}= L_{A*B}$, where we used Lemma \ref{lm:lec}, knowing that $\mathcal{S}_a\cap \mathcal{S}_b\subset \mathcal{S}_A\cap \mathcal{S}_B\subset \mathcal{S}_*$ for arbitrary $a\in A$ and $b\in B$.
\end{proof}

\begin{reptheorem}{th:cmp}
For the composition on systems $\mathcal{S}_*$, 
\begin{align}
&\mathcal{C}_{A}*\mathcal{C}_{B}\subset \mathcal{C}_{A*B}.
\end{align}
\end{reptheorem}
\begin{proof}
By (\ref{eq:cmp}), $\mathcal{C}_{A}*\mathcal{C}_{B}\subset L_A*L_B=L_{A*B}$. Intersecting with the set of positive semidefinite operators yields the result.
\end{proof}

\Cref{th:cmp} gives the best general characterization at the level of correlation constraints. The set $\mathcal{C}_{A*B}$ on the right hand side cannot be reduced, since there are $A$ and $B$ so that $\mathcal{C}_{A}*\mathcal{C}_{B}=\mathcal{C}_{A*B}$. For example, if $\mathcal{C}_{A}$ are channels from $s_1$ to $s_2$, and $\mathcal{C}_{B}$ are channels from $s_2$ to $s_3$, then $\mathcal{C}_{A}*\mathcal{C}_{B}$ composed on system $s_2$ contains all channels from $s_1$ to $s_3$, which equals $\mathcal{C}_{A*B}$ (because $A*B=\{\uh,010,110\}_{s_1s_2s_3}*\{\uh,001,011\}_{s_1s_2s_3}=\{\uh,01,11\}_{s_1s_2}$).

On the other hand, the $\subset$ cannot be replaced by $=$, either. For example, if $\mathcal{C}_{A}$ are states on $s_1$ and $\mathcal{C}_{B}$ are states on $s_2$, then $\mathcal{C}_{A}*\mathcal{C}_{B}$ composed on the trivial system are just the product states on $s_1s_2$, which is a proper subset of all the bipartite states $\mathcal{C}_{A*B}$ (because $A*B=\{\uh,10\}*\{\uh,01\}=\{\uh,10,01,11\}$).

\section{First part proof of \Cref{th:catb}}\label{sec:pthcabc}

For composition on $\mathcal{S}_*$ so that  $\mathcal{S}_*\cap \mathcal{S}_B=\varnothing$, 
\begin{align}
\mathcal{C}_{A}\rightarrow \mathcal{C}_{B}:=&\{\mathsf{C}\in\mathcal{C}:\mathcal{C}_{A}*\mathsf{C}\subset \mathcal{C}_{B}\},
\\
L_{A}\rightarrow L_{B}:=&\{\mathsf{C}\in L:L_A*\mathsf{C}\subset L_B\},
\\
A\rightarrow B:=&\{c:A*c\subset B\}
\end{align}
$\mathcal{C}_{A}\rightarrow \mathcal{C}_{B}$ is the set of processes sending $\mathcal{C}_{A}$ to $\mathcal{C}_{B}$ upon composition on $\mathcal{S}_*$. $\mathcal{S}_*\cap \mathcal{S}_B=\varnothing$ is assumed because $\mathcal{C}_{B}$, the correlators resulting from the composition, should not be supported on $\mathcal{S}_*$ where composition took place. Since composition is closed in $\mathcal{C}$, $\mathcal{C}_{A}\rightarrow \mathcal{C}_{B}$ can alternatively be written as
\begin{align}\label{eq:adctc}
\mathcal{C}_{A}\rightarrow \mathcal{C}_{B}=\{\mathsf{C}\in \mathcal{C}:\mathcal{C}_{A}*\mathsf{C}\subset L_B\}.
\end{align}
$L_{A}\rightarrow L_{B}$ is similarly the set of operators that map from $L_A$ to $L_B$. The set of correlation type elements $A\rightarrow B$ is so defined because it characterizes $\mathcal{C}_{A}\rightarrow \mathcal{C}_{B}$. We prove this in \Cref{th:cabc} below, after establishing some lemmas.

\begin{lemma}\label{lm:abnic}
Suppose $a*b\notin C$, $\mathsf{B}\in L_b$, and $\mathsf{B}\ne 0$. Then $L_a*\mathsf{B}\not\subset L_C$.
\end{lemma}
\begin{proof}
Since $a*b\notin C$, $L_{a*b}\cap L_C=\{0\}$. By  Theorem \ref{th:cmp}, $L_a*\mathsf{B}\subset L_a*L_b\subset L_{a*b}$, so $L_a*\mathsf{B}\not\subset L_C$, unless $L_a*\mathsf{B}\subset\{0\}$.

Since $a*b\notin C$, $a*b\ne \varnothing$. By the definition of $a*b$ (\Cref{cv:star}), this implies that $(a+b)_{\mathcal{S}_*}\in \{\uh,u\}$, which means either $a_{\mathcal{S}_*}=b_{\mathcal{S}_*}$, or one of them is $\uh$ and the other is $u$. In either case it is possible to pick $\mathsf{A}\in L_a$ so that $\mathsf{A}*\mathsf{B}\ne 0$. This implies $L_a*\mathsf{B}\not\subset \{0\}$.
\end{proof}

\begin{lemma}\label{lm:lccc}
For any non-zero $\mathsf{B}\in L_b$, $L_A*\mathsf{B}\subset L_C\iff A*b\subset C$.
\end{lemma}
\begin{proof}
Suppose $A*b\subset C$. Then by Theorem \ref{th:cmp}, $L_A*L_b\subset L_{A*b}\subset L_C$, which implies $L_A*\mathsf{B}\subset L_C$. Conversely, suppose $A*b\not\subset C$. Then there is an $a\in A$ so that $a*b\notin C$. Let $\mathsf{B}\in L_b$ be an arbitrary non-zero element. By Lemma \ref{lm:abnic}, $L_a*\mathsf{B}\not\subset L_C$. Hence $L_A*\mathsf{B}\not\subset L_C$.
\end{proof}

\begin{proposition}\label{th:ldm}
$L_A\rightarrow L_B=L_{A\rightarrow B}$.
\end{proposition}
\begin{proof}
Let $C=A\rightarrow B$. First, $L_A\rightarrow L_B\supset L_C$, because $L_A*L_C\subset L_{A*C}=\oplus_{c\in C}L_{A*c}\subset L_B$, where Theorem \ref{th:cmp} is used in the first step and the definition of $A\rightarrow B$ is used in the last step.

Next, we show that $L_A\rightarrow L_B\subset L_C$. Let $\mathsf{D}\in L_A\rightarrow L_B$ be arbitrary, with correlation type $D$ and expansion $\mathsf{D}=\sum_{d\in D} \mathsf{D}_d$ with $0\ne\mathsf{D}_d\in L_d$. $L_A*\mathsf{D}=\sum_{d\in D}L_A*\mathsf{D}_d\subset L_B$ implies $L_A*\mathsf{D}_d\subset L_B$ for all $d\in D$. By Lemma \ref{lm:lccc}, $A*d\subset B$ for all $d\in D$, whence $D\subset C$ and $\mathsf{D}\in L_C$. Since $\mathsf{D}$ is arbitrary, $L_A\rightarrow L_B\subset L_C$.
\end{proof}

\begin{lemma}\label{lm:lptc}
$(L_A\rightarrow L_B)\cap\mathcal{C}=\mathcal{C}_{A}\rightarrow \mathcal{C}_{B}$.
\end{lemma}
\begin{proof}
$LHS=\{\mathsf{C}\in \mathcal{C}:L_A*\mathsf{C}\subset L_B\}$, and $RHS=\{\mathsf{C}\in \mathcal{C}:\mathcal{C}_{A}*\mathsf{C}\subset L_B\}$ by (\ref{eq:adctc}). Since $\mathcal{C}_{A}\subset L_A$, $LHS\subset RHS$.

We show $LHS\supset RHS$ by contradiction. Suppose $LHS\not\supset RHS$, i.e., there exists $\mathsf{C}\in RHS$ so that $\mathsf{C}\notin LHS$. Since $\mathsf{C}
\in \mathcal{C}$, this implies that $\mathsf{C}\notin L_A\rightarrow L_B=L_D$, where by \Cref{th:ldm}, $D=\{d:A*d\subset B\}$. Denoting the correlation type of $\mathsf{C}$ by $C$, we can find $c\in C$ so that $c\notin D$, which implies there is $a\in A$ so that $a*c\notin B$. By Lemma \ref{lm:abnic}, $L_a*\mathsf{C}\not\subset L_B$, which implies we can find $\mathsf{A}\in L_a$ so that $\mathsf{A}*\mathsf{C}\notin L_B$. By adding a multiple of $\id$, we can make $\mathsf{A}\in \mathcal{C}_{A}$ so that $\mathsf{A}*\mathsf{C}\notin \mathcal{C}_{B}$. This contradicts the assumption that $\mathsf{C}\in RHS$.
\end{proof}

\begin{proposition}\label{th:cabc}
For composition on $\mathcal{S}_*$ so that  $\mathcal{S}_*\cap \mathcal{S}_B=\varnothing$,
\begin{align}
\mathcal{C}_{A}\rightarrow \mathcal{C}_{B}=\mathcal{C}_{A\rightarrow B}.
\end{align}
\end{proposition}
\begin{proof}
By  Lemma \ref{lm:lptc} and \Cref{th:ldm}, $LHS=(L_A\rightarrow L_B)\cap \mathcal{C}=L_{A\rightarrow B}\cap \mathcal{C}=RHS$.
\end{proof}
This establishes the first part (\ref{eq:catbc}) of \Cref{th:catb}.

\section{Second part proof of \Cref{th:catb}}\label{sec:pcatb}

Suppose $\mathcal{S}_A\subset \mathcal{S}$. By the definition of $\pp{A_\mathcal{S}}$ (\ref{eq:perp}) and the convention that $\uh\in A$ implicitly whenever $u\in A$,
\begin{align}
(\pp{A_\mathcal{S}})_\mathcal{S}^\perp=
\begin{cases}\label{eq:app}
A\backslash\{\uh\}, \quad &\uh\ein A,
\\
A, \quad &\text{otherwise}.
\end{cases}
\end{align}

By (\ref{eq:apb}) and (\ref{eq:amb}) we have $a+b=a-b$, except when $a=\uh, b\ne\uh$. Moreover (recall \Cref{cv:cts} for $\emptyset$),
\begin{align}\label{eq:apbmb}
(a+b)-b=
\begin{cases}
u,\quad &a=\uh,b\ne\uh,
\\
a,\quad &\text{ otherwise}.
\end{cases}
\end{align}
\begin{align}\label{eq:ambpb}
(a-b)+b=
\begin{cases}
\emptyset,\quad &a=\uh,b\ne\uh,
\\
a,\quad &\text{ otherwise}.
\end{cases}
\end{align}

\begin{lemma}
Suppose $\mathcal{S}_B\cap \mathcal{S}_*=\varnothing$. Then
\begin{align}\label{eq:atbp}
  \pp{(A\rightarrow B)}=\pp{B}-A  
\end{align}
\end{lemma}
\begin{proof}
For simplicity denote $A\rightarrow B=\{c:A*c\subset B\}$ by $C$. First we show that $\pp{C}\subset\pp{B}-A$. If $C=\varnothing$ the statement clearly holds. Otherwise let $c'\in \pp{C}$ be arbitrary. It suffices to show that $c'\in \pp{B}-A$. By the definition of $C$, there exists $a\in A$ so that $a*c'\notin B$. This implies that $b':=a+c'=a*c'$. Since $c'\in \pp{C}$, $c'\ne \uh$, whence $b'\ne \uh$. This implies $b'\in \pp{B}$. Now $c'\ne \uh$ and (\ref{eq:apbmb}) imply that $c'=(a+c')-a=b'-a\in \pp{B}-A$.

Next we show that $\pp{C}\supset\pp{B}-A$. It suffices to show that $b'-a\in \pp{C}$ for arbitrary $b'\in \pp{B}$ and $a\in A$. Since $b'\in \pp{B}$, $b'\ne \uh$. By (\ref{eq:ambpb}), $(b'-a)+a=b'\in \pp{B}$. Since $b'\in \pp{B}$ has support within $\mathcal{S}_B$ and $\mathcal{S}_B\cap \mathcal{S}_*=\varnothing$, $b'_{\mathcal{S}_*}=((b'-a)+a)_{\mathcal{S}_*}\in \{\uh,u\}$. Then $(b'-a)*a=(b'-a)+a=b'\in \pp{B}$, whence $b'-a\notin C$. Since $b'\ne \uh$, $b'-a\ne\uh$, which means $b'-a\in \pp{C}$.
\end{proof}

\begin{proposition}
\begin{align}
A\rightarrow B=
\begin{cases}
(\pp{B}-A)_{\mathcal{S}_{A\rightarrow B}}^\perp\cup\{\uh\}, \quad &A*\uh\subset B,
\\
(\pp{B}-A)_{\mathcal{S}_{A\rightarrow B}}^\perp, \quad &\text{otherwise}.
\end{cases}
\end{align}
\end{proposition}
\begin{proof}
This is a direct consequence of (\ref{eq:app}) and (\ref{eq:atbp}).
\end{proof}
This establishes the second part (\ref{eq:catb}) of \Cref{th:catb}.

\section{Third part proof of \Cref{th:catb}}\label{sec:pcme}

Given $A$, $B$ and $\mathcal{S}_*$ disjoint from $\mathcal{S}_B$, it is not guaranteed that there are non-zero elements in $\mathcal{C}_{A}\rightarrow \mathcal{C}_{B}$, even when $A\rightarrow B\ne \varnothing$. 
\begin{example}
Let $\mathcal{S}=\{s_1,s_2,s_3\}, \mathcal{S}_A=\{s_2,s_3\}$, $\mathcal{S}_B=\{s_1,s_3\}$, $\mathcal{S}_*=\{s_2\}$, $A=\{\uh,001,011\}$, and $B=\{\uh,101\}$. Then $\mathcal{C}_{A}\rightarrow \mathcal{C}_{B}=\{0\}$.
\end{example}
\noindent By (\ref{eq:catb}), $A\rightarrow B=\{110\}$. Yet $\mathcal{C}_{A}\rightarrow \mathcal{C}_{B}=\mathcal{C}_{A\rightarrow B}=L_{A\rightarrow B}\cap \mathcal{C}=\{0\}$. 
The reason that $A\rightarrow B=\{110\}$ does not yield non-zero processes is the following.
\begin{lemma}\label{lm:ctuuh} 
Let $\mathsf{A}\ne 0$ be a positive semidefinite operator with correlation type $A$. Either $u\in A$ (implying $\uh \in A$ by our convention) or $\uh\ein A$.
\end{lemma}
\begin{proof}
By assumption, $\mathsf{A}>0$, so $\Tr \mathsf{A}>0$. Among all the correlation type elements, only $u$ and $\uh$ supply positive trace, so the result follows.
\end{proof}

\begin{proposition}\label{th:cme}
$\mathcal{C}_{A}\rightarrow \mathcal{C}_{B}\ne \{0\}$ if and only if $A*\uh\subset B$.
\end{proposition}
\begin{proof}
$\mathcal{C}_{A}\rightarrow \mathcal{C}_{B}=\mathcal{C}_{A\rightarrow B}$ by \Cref{th:cabc}. Suppose $\mathcal{C}_{A\rightarrow B}=\mathcal{C}_{A}\rightarrow \mathcal{C}_{B}\ne \{0\}$. By Lemma \ref{lm:ctuuh}, $\uh\in A\rightarrow B$, and $A*\uh\subset B$. Now suppose $A*\uh\subset B$. Then $\uh \in A\rightarrow B$, and $\id \in \mathcal{C}_{A\rightarrow B}=\mathcal{C}_{A}\rightarrow \mathcal{C}_{B}$.
\end{proof}
This establishes the third part (\ref{eq:ec}) of \Cref{th:catb}.

\bibliographystyle{unsrt}
\bibliography{mendeley.bib}
\end{document}